\documentclass[sigconf]{acmart}

\makeatletter
\def\@ACM@checkaffil{% Only warnings
    \if@ACM@instpresent\else
    \ClassWarningNoLine{\@classname}{No institution present for an affiliation}%
    \fi
    \if@ACM@citypresent\else
    \ClassWarningNoLine{\@classname}{No city present for an affiliation}%
    \fi
    \if@ACM@countrypresent\else
        \ClassWarningNoLine{\@classname}{No country present for an affiliation}%
    \fi
}
\makeatother

\AtBeginDocument{%
  \providecommand\BibTeX{{%
    \normalfont B\kern-0.5em{\scshape i\kern-0.25em b}\kern-0.8em\TeX}}}

\copyrightyear{2023}
\acmYear{2023}
\setcopyright{rightsretained}
\acmConference[SIGIR '23]{Proceedings of the 46th International ACM SIGIR Conference on Research and Development in Information Retrieval}{July 23--27, 2023}{Taipei, Taiwan}
\acmBooktitle{Proceedings of the 46th International ACM SIGIR Conference on Research and Development in Information Retrieval (SIGIR '23), July 23--27, 2023, Taipei, Taiwan}
\acmPrice{}
\acmDOI{10.1145/3539618.3591947}
\acmISBN{978-1-4503-9408-6/23/07}

\begin{document}

\title[Popularity Debiasing from Exposure to Interaction in Collaborative Filtering]{Popularity Debiasing from Exposure to Interaction in Collaborative Filtering}

\author{Yuanhao Liu}
\affiliation{%
  \institution{Data Intelligence System Research Center, Institute of Computing
Technology, CAS}
}
\affiliation{%
  \institution{University of Chinese Academy of Sciences, Beijing, China}
}
\email{liuyuanhao20z@ict.ac.cn}

\author{Qi Cao}
\authornote{Corresponding authors.}
\affiliation{%
  \institution{Data Intelligence System Research Center, Institute of Computing
Technology, CAS}
}
\email{caoqi@ict.ac.cn}

\author{Huawei Shen}
\authornotemark[1]
\affiliation{%
  \institution{Data Intelligence System Research Center, Institute of Computing
Technology, CAS}
}
\affiliation{%
  \institution{University of Chinese Academy of Sciences, Beijing, China}
}
\email{shenhuawei@ict.ac.cn}

\author{Yunfan Wu}
\affiliation{%
  \institution{Data Intelligence System Research Center, Institute of Computing
Technology, CAS}
}
\affiliation{%
  \institution{University of Chinese Academy of Sciences, Beijing, China}
}
\email{wuyunfan19b@ict.ac.cn}

\author{Shuchang Tao}
\affiliation{%
  \institution{Data Intelligence System Research Center, Institute of Computing
Technology, CAS}
}
\affiliation{%
  \institution{University of Chinese Academy of Sciences, Beijing, China}
}
\email{taoshuchang18z@ict.ac.cn}

\author{Xueqi Cheng}
\affiliation{%
  \institution{CAS Key Lab of Network Data Science and Technology, Institute of Computing Technology, CAS}
}
\affiliation{%
  \institution{University of Chinese Academy of Sciences, Beijing, China}
}
\email{cxq@ict.ac.cn}
\begin{abstract}
Recommender systems often suffer from popularity bias, where popular items are overly recommended while sacrificing unpopular items. Existing researches generally focus on ensuring the number of recommendations (\emph{exposure}) of each item is equal or proportional, using inverse propensity weighting, causal intervention, or adversarial training.
However, increasing the exposure of unpopular items may not bring more clicks or interactions, resulting in skewed benefits and failing in achieving real reasonable popularity debiasing. 
In this paper, we propose a new criterion for popularity debiasing, i.e., \emph{in an unbiased recommender system, both popular and unpopular items should receive \underline{I}nteractions \underline{P}roportional to the number of users who \underline{L}ike it, namely IPL criterion}. Under the guidance of the criterion, we then propose a debiasing framework with IPL regularization term which is theoretically shown to achieve a win-win situation of both popularity debiasing and recommendation performance. Experiments conducted on four public datasets demonstrate that when equipping two representative collaborative filtering models with our framework, the popularity bias is effectively alleviated while maintaining the recommendation performance.
\end{abstract}




\maketitle

\section{Introduction}
Recommender systems aim to help users find interesting items and are used in various online services, e.g., online shopping \cite{wu2019neural, kim2002personalized, kim2009personalized}, social networks \cite{debnath2008feature, walter2008model, wang2012joint}, and video sites \cite{davidson2010youtube, deldjoo2016content, mei2011contextual}.
Collaborative filtering, a widely used recommendation approach~\cite{koren2009matrix, he2020lightgcn, wu2016collaborative, huang2023adversarial}, learns models for recommendation from user-item interactions.
Unfortunately, since the user-item interactions usually exhibit long-tail distribution in terms of item popularity~\cite{zhang2021model, li2019both}, the models trained via the skewed interactions are often observed to further amplify the popularity bias by over-recommending popular items and under-recommending unpopular items~\cite{mansoury2020feedback, abdollahpouri2019unfairness}. Such popularity bias may lead to terrible results such as Matthew Effect~\cite{perc2014matthew}, hindering the long-term and sustainable development of recommender systems.

Researchers propose different methods for popularity debiasing.
Early works generally follow a paradigm of inverse propensity weighting~\cite{kamishima2014correcting, schnabel2016recommendations, liang2016causal, brynjolfsson2006niches, krishnan2018adversarial}, i.e., increase the weight of unpopular items during training.
Later, researchers introduce causal methods to remove the harmful impact of popularity~\cite{zhang2021causal, zheng2021disentangling, wei2021model, zhao2021popularity}. Adversarial learning is also used to enhance the distribution similarity of prediction scores between items to alleviate popularity bias~\cite{zhu2020measuring, liu2022mitigating}.

\begin{figure}
    \centering
    \includegraphics[width=7.3cm]{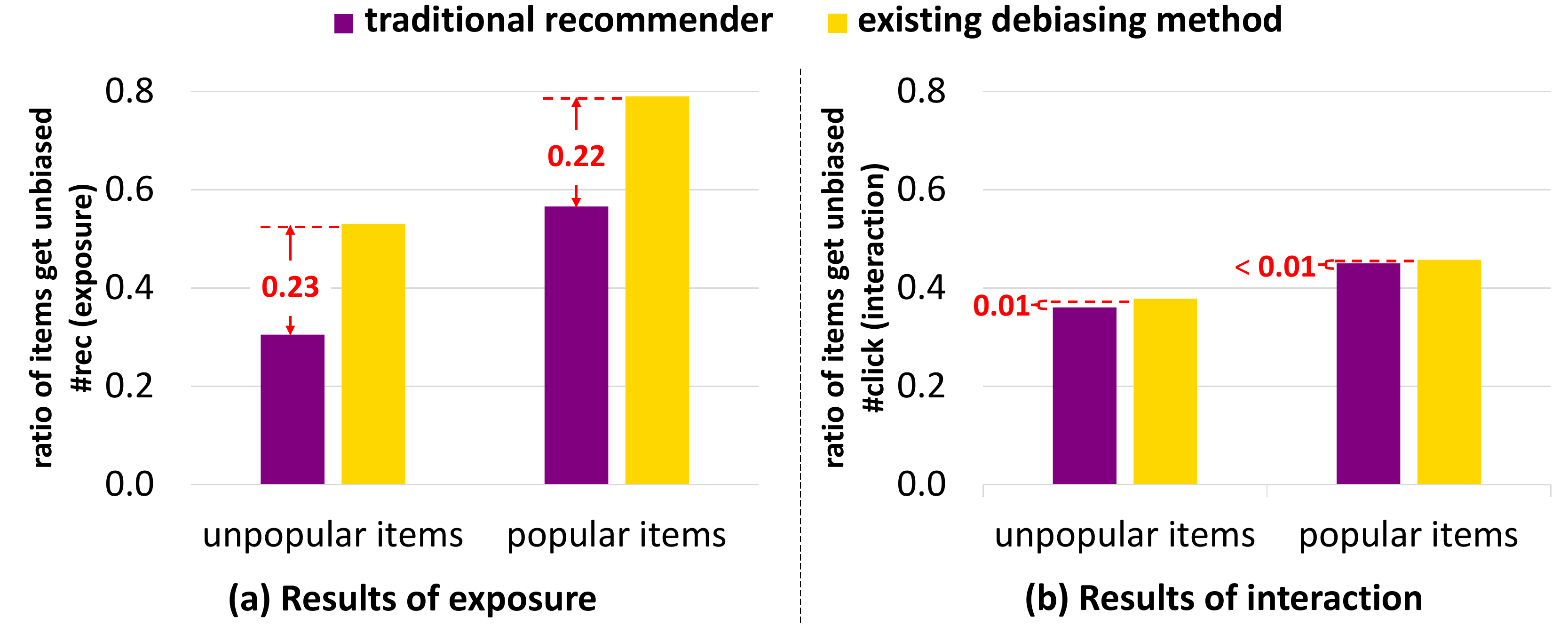}
    \setlength{\abovecaptionskip}{0.2cm}
    \caption{The proportion of items that get the unbiased (a) \textit{exposure} (b) \textit{interaction} amount before and after debiasing. The unbiased number of exposures/interactions of the item is proportional to the item's popularity. Popular items are the ones whose popularity exceeds 80\% items.}
    \label{fig:fig1}
    \vspace{-5pt}
\end{figure}

Despite the various existing methods, they generally focus on making the number of recommendations (the \emph{exposure}) of items equal to each other or proportional to their popularity~\cite{wei2021model, liang2016causal}. 
Fig.~\ref{fig:fig1}(a) shows the ratio of items that get unbiased exposure amount in existing representative debiasing method~\cite{wei2021model}, which helps 23\% unpopular items avoid being under-recommended. 
However, blindly increasing the exposure of unpopular items may not significantly increase their interactions
\footnote{User \emph{interactions}, such as clicking, purchasing, or expressing liking for the recommended item, can directly represent that the item benefits from the recommendation.}, since they may be recommended to users who dislike them. 
Taking two similar books as an example, one is recommended to users who like it and the other is recommended more to those who don't like it. 
Although they are treated similarly in exposure, the latter is likely to earn less because fewer users will buy it, and thus still be plagued by popularity bias. 
As shown in Fig.~\ref{fig:fig1}(b), there are hardly any extra unpopular items (less than 0.01) that can get an unbiased amount of interactions through existing debiasing methods.
In other words, simply increasing the exposure of unpopular items in existing methods may not necessarily lead to more clicks or interactions, resulting in skewed benefits and failing in achieving real reasonable popularity debiasing.

In this paper, we propose a new criterion for popularity debiasing, i.e., in a popularity unbiased recommender system, both popular and unpopular items should receive \emph{interactions} \emph{proportional} to the number of users who \emph{like} it, namely \textit{\textbf{IPL}} criterion. We propose a regularization-based debiasing framework following this criterion. However, in offline scenarios, we cannot observe the number of interactions that items can get in recommendation. In addition, due to the missing-not-at-random phenomenon of data, users who like the item are also invisible. In order to estimate both the number of interactions and users who like the items using observational data, we propose an inverse-propensity-score-based estimation method and use it to formalize the IPL regularization term for debiasing.

We conduct experiments on four public datasets and equip our framework with the representative matrix factorization~\cite{koren2009matrix} and the state-of-the-art LightGCN~\cite{he2020lightgcn} for collaborative filtering. Experimental results show that our framework achieves the best effectiveness of debiasing while even further improving the recommendation performance. Further analysis is also offered to deepen the understanding of the success of our framework.

In summary, the contribution of this work is threefold.

\begin{itemize}{\leftmargin=1em}
\setlength{\topmargin}{0pt}
\setlength{\itemsep}{0em}
\setlength{\parskip}{0pt}
\setlength{\parsep}{0pt}
    \item We propose a new criterion for popularity debiasing, i.e., an item should receive interactions proportional to the number of users who like it.
    \item We propose an estimation method to measure IPL in offline data, and further propose a debiasing framework, which is theoretically shown to achieve a win-win situation of both popularity debiasing and recommendation performance. 
    \item Experiments conducted on four public datasets demonstrate the effectiveness of our framework when equipped with two representative collaborative filtering models.
\end{itemize}
\vspace{-5pt}
\section{Method}
This section describes the problem formulation and the proposed new criterion for popularity debiasing, as well as a general debiasing framework.
% A theoretical analysis is further offered to show the compatibility of the criterion with recommendation performance.
\vspace{-5pt}

\subsection{Collaborative Filtering}
Suppose we have $N$ users $\mathcal{U}=\{u_1, u_2, ..., u_N\}$ and $M$ items $\mathcal{I}=\{i_1, i_2, ..., i_M\}$, collaborative filtering aims to predict the preference score $\widehat{y}_{u, i}$ of item $i$ by user $u$ based on observed interactions. For training a collaborative filtering model, Bayesian Personalized Ranking (BPR) \cite{rendle2009bpr} is one of the most widely adopted training loss,
\begin{equation*}
\small
\vspace{-3pt}
    L_{\text{BPR}}=-\sum_{u\in \mathcal{U}}\sum_{\substack{i^+\in \mathcal{I}_u^+}}\sum_{i^-\in \mathcal{I}/\mathcal{I}_u^+} \ln(\sigma(\widehat{y}_{u,i^+} - \widehat{y}_{u, i^-})+ \lambda_{\Theta}\lVert\Theta\rVert^2_F,
\vspace{-3pt}
\end{equation*}
where $\mathcal{I}_u^+$ is the set of items interacted with $u$ in training data.

\vspace{-5pt}
\subsection{New Criterion for Popularity Debiasing} \label{definition}
The new criterion is defined as a popular or unpopular item should receive interactions proportional to the number of users who like it, namely IPL criterion, which is formally defined as:
\begin{equation}\label{criterion}
\small
\vspace{-3pt}
\forall i, j \in \mathcal{I}^{\text{pop}} or \mathcal{I}^{\text{unpop}}, \frac{C_i}{Q_i}=\frac{C_j}{Q_j},
% \vspace{-3pt}
\end{equation}
where $\mathcal{I}^{pop}$ and $\mathcal{I}^{unpop}$ are popular and unpopular items respectively, $C_i$ is the number of interactions $i$ can get from recommendation, and $Q_i$ is the number of interactions that $i$ can get when exposed to all users, reflecting the number of users who like item $i$. Such a criterion takes into account two aspects:

\begin{itemize}{\leftmargin=1em}
\setlength{\topmargin}{0pt}
\setlength{\itemsep}{0em}
\setlength{\parskip}{0pt}
\setlength{\parsep}{0pt}
    \item The benefit of an item in recommendation is more about the interactions it receives, instead of massive useless recommendations. Hence, we take the interactions as the indicator.
    \item The more users like an item, the more interactions it deserves. Hence, an unbiased recommendation should ensure interactions is proportional to the users who like the item.
\end{itemize}

A recommender system is ideally unbiased if it satisfies the IPL criterion above, which not only provides a more reasonable goal for debiasing but also eliminates the trade-off between debiasing and recommendation performance (see Section 2.5).
\vspace{-5pt}
\subsection{Estimation in offline data}

The IPL criterion provides a great goal for popularity debiasing. However, due to the incomplete observation in offline data, it is hard to obtain the real number of interactions received from recommendations ($C_i$) and from all users ($Q_i$). Although some methods tried to consider interaction amount in debiasing~\cite{zhu2020measuring, zhu2021popularity}, they generally obtain the $C_i$ directly in the online scenario or ignore the incompleteness of the data, making it difficult to apply offline.

To solve the above problem, we propose an estimator for $C_i$ and $Q_i$ based on the inverse-propensity-score (IPS)~\cite{joachims2016counterfactual} technique.

For estimating $C_i$, note that $C_i = \sum_{u\in \mathcal{U}_i}R_{u, i}$, where $\mathcal{U}_i$ is users interact with $i$ when $i$ is exposed to all users, and $R_{u, i}=1$ if $i$ is recommended to $u$ otherwise $0$. Given $C_i^{\text{IPS}} = \sum_{u\in \mathcal{U}_i}\frac{O_{u, i}}{P_{u,i}}R_{u,i}$ where $P_{u,i} = \text{Pr}(O_{u, i}=1)$ and $O_{u, i}$ is a binary variable of which $1$ means that an interaction between $u$ and $i$ is observed, we can easily have $\mathbb{E}_O\left[C_i^{\text{IPS}}\right] = C_i$. In other words, $C_i^{\text{IPS}}$ is an unbiased estimation of $C_i$. For estimating $P_{u, i}$ in $C_i^{\text{IPS}}$, we follow the broadly adopted approach in~\cite{yang2018unbiased}, which assume that $P_{u, i}=P_{*, i}=p^{\text{expose}}_{*, i}*p^{\text{like}}_{*, i}$, where $p^{\text{expose}}_{*, i}$ and $p^{\text{like}}_{*, i}$ are the probability that item $i$ is exposed and liked, and $p^{\text{like}}_{*, i}$ is estimated to be proportional to $Q_i=|\mathcal{U}_i|$. Then we have $P_{*, i}\propto p_{*,i}^{\text{expose}}*Q_i$. Denoting the observed popularity of item $i$ as $Q_i^*=|\mathcal{U}_i^*|$, we have a binomial distribution of $Q_i^*\sim \mathcal{B}(Q_i, p_{*,i}^{\text{expose}})$. Then we have $Q_i^*=p_{*,i}^{\text{expose}}*Q_i\propto P_{*, i}$, and
\begin{equation}
    \small
    \vspace{-5pt}
    C_i^{\text{IPS}}
    =\sum_{u\in \mathcal{U}_i}\frac{O_{u, i}}{P_{u,i}}R_{u,i}
    =\sum_{u\in \mathcal{U}_i^*}\frac{R_{u,i}}{P_{u,i}}
    =\frac{\sum_{u\in \mathcal{U}_i^*}R_{u,i}}{P_{*,i}}\propto\frac{C_i^*}{Q_i^*},
\end{equation}
where $C_i^*=\sum_{u\in \mathcal{U}_i^*}R_{u,i}$ is the number of observed interactions in recommendation.

To estimate $Q_i$, recall that $Q_i=Q_i^*/p_{*,i}^{\text{expose}}$. We assume that $p_{*,i}^{\text{expose}}\!\propto\!\left(Q_i^*\right)^\gamma$~\cite{yang2018unbiased}, where $\gamma$ characterizes the popularity bias during data generation~\cite{steck2011item}, and can be estimated by fitting the data. Then we have
\begin{equation}
    \small
    \vspace{-5pt}
    Q_i=\frac{Q_i^*}{p_{*, i}^{\text{expose}}}\propto{(Q_i^*)}^{1-\gamma}.
\end{equation}

With observed $C_i^*$ and $Q_i^*$ in offline data, we can measure the IPL criterion by $\frac{C_i^{\text{IPS}}}{Q_i} \propto \frac{C_i^*}{\left(Q_i^*\right)^{2-\gamma}} = r_i$. Then we can rewrite Eq. (\ref{criterion}) as:
\begin{equation}\label{measuring}
    \small
    \vspace{-5pt}
    \forall i, j \in \mathcal{I}^{\text{pop}} or \mathcal{I}^{\text{unpop}}, r_i = r_j.
    \vspace{-5pt}
\end{equation}
\vspace{-5pt}
\subsection{Debiasing Framework}
We further propose a debiasing framework with a regularization term considering the IPL criterion during model training. 

Considering a recommender model with parameter $\Theta$, we have
\begin{equation*}
\small
    \begin{aligned}
    \widehat{r}_i
    =\mathbb{E}[r_i|\Theta]
    =\frac{\mathbb{E}\left[\left.C_i^*\right| \Theta \right]}{\left(Q_i^*\right)^{2-\gamma}} 
    =\frac{\sum_{u \in \mathcal{U}^*_i}\text{Pr}\left(\left.u,i\right|\Theta\right)}{\left|\mathcal{U}^*_i\right|^{2-\gamma}}
    =\frac{\sum_{u \in \mathcal{U}^*_i}\sigma(\widehat{y}_{u,i})}{\left|\mathcal{U}^*_i\right|^{2-\gamma}},
    \end{aligned}
\end{equation*}
where $\sigma(\widehat{y}_{u,i})$ indicates the probability that $i$ will be recommended to $u$, $\mathcal{U}_i^*$ is all the users observed interact with $i$ in training data.

Then the regularization term ensures IPL criterion is written as:
\begin{equation}
\small
\vspace{-3pt}
L_{\text{IPL}}
=std_{i\in\mathcal{I}}(\widehat{r}_i)
=\sqrt{\frac{1}{M}\sum_{i \in \mathcal{I}}(\widehat{r}_i - \frac{1}{M}\sum_{i\in \mathcal{I}}\widehat{r}_i)^2},
\end{equation}
which aims to reduce the $r_i$ gap between all the $M$ items according to Eq. (\ref{measuring}) to achieve unbiased recommendation under the IPL criterion.

The final loss function for debiasing model training is:
\begin{equation}
\small
    \min_\Theta L_{\text{debias}}=L_{\text{BPR}} + \lambda_f*L_{\text{IPL}},
    \label{fairloss}
\end{equation}
where $\lambda_f$ controls the strength of popularity debiasing.
\vspace{-5pt}
 \subsection{Theoretical Analysis}\label{theoretical}

Generally, popularity debiasing may bring a decrease in recommendation performance. In this section, we give a theoretical analysis to show that our proposed criterion can achieve a win-win situation for both recommendation performance and popularity debiasing. 

\begin{proposition}\label{proposition}
Given an acceptable recommendation performance with $\text{recall}=c$, there always exists a recommendation result that satisfies the requirements of both the recall $c$ and the IPL criterion.
\end{proposition}

\begin{proof}
For recommendations with recall $c = \frac{\sum_{i\in \mathcal{I}}C_i}{\sum_{i\in \mathcal{I}}Q_i}$, to meet the IPL criterion, we should have $\forall i \in \mathcal{I}, r_i=\frac{C_i}{Q_i}=\frac{\sum_{i\in \mathcal{I}}C_i}{\sum_{i\in \mathcal{I}}Q_i}=c$.

Given $k \in \{1, 2, ..., M\}$ denoting the number of items recommended to each user, if a recommendation with recall $c$ cannot meet the IPL criterion, then there always exists a discriminated item $i$ that has $r_i=\frac{C_i}{Q_i}=\frac{C_i}{|\mathcal{U}_i|}<c$. In other words, there are at least $(1-c)|\mathcal{U}_i|$ users like item $i$ while cannot further have $i$ in their recommendation list, which means that such a user $u$ will necessarily lose the recommendation of another item $j\in\mathcal{I}_u$ ($u$ likes $j$) and make $r_j<c$ if adding item $i$ into the recommendation list of user $u$. For better understanding, we denote the set of items at risk of discrimination (like $j$) as $\mathcal{I}^-{\triangleq} \left\{j\left|\frac{|\mathcal{U}^- \cap \mathcal{U}_j|}{|U_j|} > 1-c\right. \right\}$, where $\mathcal{U}^-{\triangleq}\left\{u\big||\mathcal{I}_u|>k\right\}$. The above-mentioned user $u$ likes at least $k$ items in $\mathcal{I}^-$. Now we consider the existence of $u$, which is the necessary condition for the existence of the discriminated item $i$:
\begin{equation}
\small
 \textbf{condition-1}: \exists u \in \mathcal{U}, |\mathcal{I}_u \cap \mathcal{I}^-| > k.
 \vspace{-3px}
\end{equation}

Given $|\mathcal{I}_u|$ follows Pareto distribution $\text{Pr}(|\mathcal{I}_u|\!=\!x)\!=\!(\beta-1)x^{-\beta}$ with $\beta>1$, the probability that user $u\in\mathcal{U}^-$ equals to $p=\text{Pr}(|\mathcal{I}_u|\!>\!k)=k^{1-\beta}$. Therefore we have $\text{Pr}(|\mathcal{U}^- \!\cap\! \mathcal{U}_i|\!=\!x) \sim B(x; |\mathcal{U}_i|, p)$, where $B(\cdot;n, p)$ is the Binomial distribution with the number of experiments $n$ and the probability of success $p$. Then we have $\text{Pr}(i\in\mathcal{I}^-)=\text{Pr}\left(|\mathcal{U}^- \cap \mathcal{U}_i| \!>\! (1-c)|\mathcal{U}_i|\right)< F(c|\mathcal{U}_i|; |\mathcal{U}_i|, 1-p)$, where $F(\cdot;n, p)$ is the cumulative distribution function of $B(\cdot;n, p)$. With the Chernoff bound $F(x; n, p)\leq \exp\left(-nD(\frac{x}{n}||p)\right)$, where $D(a||p)=a\log \frac{a}{p} + (1-a)\log \frac{1-a}{1-p}$, we further have $\text{Pr}(i\in\mathcal{I}^-)\leq\exp\left(-\left((1-c)\log\frac{(1-c)}{p}+c\log\frac{c}{1-p}\right)\right)$.

Then the probability that satisfying \textit{condition-1} is:
\begin{equation}\small
\begin{aligned}
&1{-}\prod_{u\in\mathcal{U}\atop{|\mathcal{I}_u|>k}}\left(
        1 {-} \sum_{j=k+1}^{|\mathcal{I}_u|}\tbinom{|\mathcal{I}_u|}{j}\text{Pr}(i\in\mathcal{I}^-)^j(1-\text{Pr}(i\in\mathcal{I}^-))^{|\mathcal{I}_u| - j}
    \right)\\
    \leq &1{-}\prod_{u\in\mathcal{U}\atop{|\mathcal{I}_u|>k}}\left(
        1 {-}\sum_{j=k+1}^{|\mathcal{I}_u|}\tbinom{|\mathcal{I}_u|}{j}\text{Pr}(i\in\mathcal{I}^-)^j
    \right)\\
    \leq &1{-}\prod_{u\in\mathcal{U}\atop{|\mathcal{I}_u|>k}}\left(
        1 {-}\sum_{j=k+1}^{|\mathcal{I}_u|}\tbinom{|\mathcal{I}_u|}{j}\exp\left(-\left((1{-}c)\log\frac{(1{-}c)}{p}{+}c\log\frac{c}{1{-}p}\right)\right)^j
    \right)
\end{aligned}
\vspace{-3px}
\end{equation}

Taking real-world datasets as an illustration, when specifying interaction recall $c=0.99$, the probability of \textit{condition-1} $<10^{-10}$ on MovieLens-1M. Similar results were shown on other datasets. This show that there is almost always a recommendation result that satisfies IPL criterion while ensuring recommendation performance.
\end{proof}

% Proposition \ref{proposition} gives us a theoretical guarantee that the IPL criterion can achieve a win-win situation for both debiasing and recommendation performance.
\vspace{-5pt}
\section{Experimental results}

\begin{table*}[htbp]
    \setlength{\abovecaptionskip}{0cm}  %段前
    \setlength{\belowcaptionskip}{-0.2cm} %段后
    \centering
    \caption{Performance of recommendation and debiasing. Bold and underlined scores are the best and second-best performances.}
    \resizebox{0.99\textwidth}{!}{
        \begin{tabular}{l|ccccc|ccccc|ccccc|ccccc}
            \toprule
            & \multicolumn{5}{c|}{MovieLens-1M}             & \multicolumn{5}{c|}{Gowalla}          & \multicolumn{5}{c|}{Yelp}             & \multicolumn{5}{c}{Amazon book} \\
            \cmidrule{2-21} & Pre$\uparrow$   & Recall$\uparrow$ & NDCG$\uparrow$  & MI$\downarrow$  & DI$\downarrow$ & Pre$\uparrow$   & Recall$\uparrow$ & NDCG$\uparrow$  & MI$\downarrow$   & DI$\downarrow$ & Pre$\uparrow$   & Recall$\uparrow$ & NDCG$\uparrow$  & MI$\downarrow$  &   DI$\downarrow$ & Pre$\uparrow$ & Recall$\uparrow$ & NDCG$\uparrow$  & MI$\downarrow$  & DI$\downarrow$ \\
            \midrule
            MF    & \underline{19.977 } & \underline{27.685 } & \underline{30.777 } & 2.138  & 3.723  & \underline{4.738 } & \underline{17.763 } & \underline{13.149 } & 0.601  & 3.067  & \underline{2.557 } & \underline{10.871 } & \underline{6.848 } & 0.782  & 3.989  & \underline{2.618 } & \underline{12.066 } & \underline{7.979 } & 0.757  & 3.061  \\
            PO    & 17.329  & 23.097  & 25.877  & 1.997  & 3.461  & 4.442  & 16.862  & 12.374  & 0.593  & 2.918  & 1.951  & 8.075  & 5.078  & 0.742  & 4.594  & 2.606  & 12.021  & 7.965  & 0.742  & \underline{2.928 } \\
            DPR   & 16.630  & 23.010  & 25.240  & \underline{1.985 } & 3.467  & 3.983  & 14.404  & 10.862  & 0.611  & 4.042  & 2.104  & 8.485  & 5.407  & 0.766  & 5.757  & 2.238  & 10.219  & 7.596  & 0.728  & 5.194  \\
            IPW   & 16.350  & 22.110  & 23.540  & 2.073  & 3.500  & 4.021  & 15.151  & 10.360  & 0.588  & 2.974  & 1.950  & 8.081  & 4.813  & \underline{0.718 } & \underline{3.882 } & 2.220  & 8.751  & 7.011  & 0.735  & 3.004  \\
            MACR  & 17.248  & 23.253  & 25.782  & 2.211  & \underline{3.458 } & 3.977  & 14.694  & 10.879  & \underline{0.584 } & \underline{2.905 } & 1.685  & 6.732  & 4.285  & 0.747  & 4.891  & 2.140  & 11.875  & 7.889  & \underline{0.726 } & 3.173  \\
            \textbf{IPL} & \textbf{20.397 } & \textbf{28.642 } & \textbf{31.397 } & \textbf{1.942 } & \textbf{3.254 } & \textbf{5.136 } & \textbf{19.220 } & \textbf{14.450 } & \textbf{0.559 } & \textbf{2.882 } & \textbf{2.715 } & \textbf{11.578 } & \textbf{7.278 } & \textbf{0.625 } & \textbf{3.808 } & \textbf{2.891 } & \textbf{13.274 } & \textbf{8.914 } & \textbf{0.725 } & \textbf{2.889 } \\
            \midrule
            LGCN  & 19.808  & \underline{27.217 } & \underline{30.599 } & 2.010  & 4.032  & \underline{5.381 } & \underline{19.894 } & \underline{15.355 } & 0.719  & \underline{2.623 } & \underline{2.934 } & \underline{12.410 } & \underline{7.948 } & 0.723  & \underline{4.403 } & \underline{2.905 } & \underline{12.978 } & \underline{8.805 } & 0.777  & \underline{3.016 } \\
            PO    & \underline{19.830 } & 27.206  & 30.475  & 1.998  & 3.356  & 5.221  & 18.373  & 15.176  & 0.692  & 3.131  & 2.861  & 12.220  & 7.767  & 0.738  & 4.487  & 2.880  & 11.675  & 7.992  & 0.720  & 3.031  \\
            DPR   & 17.682  & 22.861  & 28.199  & \underline{1.922 } & 3.417  & 4.533  & 16.654  & 12.458  & 0.638  & 2.760  & 2.106  & 11.255  & 6.152  & 0.728  & 6.809  & 2.422  & 9.427  & 6.854  & 0.736  & 4.334  \\
            IPW   & 13.689  & 16.700  & 19.590  & 1.959  & 3.777  & 2.950  & 10.960  & 8.340  & 0.685  & 3.357  & 1.913  & 9.159  & 5.822  & \underline{0.722 } & 6.026  & 2.075  & 8.427  & 6.999  & \underline{0.716 } & 4.922  \\
            MACR  & 18.724  & 26.251  & 27.945  & 1.985  & \underline{3.180 } & 5.183  & 19.022  & 14.831  & \underline{0.630 } & 3.058  & 2.335  & 11.684  & 7.702  & 0.725  & 4.651  & 2.108  & 10.133  & 7.255  & 0.757  & 3.799  \\
            \textbf{IPL} & \textbf{19.979 } & \textbf{27.545 } & \textbf{30.949 } & \textbf{1.921 } & \textbf{3.005 } & \textbf{5.385 } & \textbf{19.897 } & \textbf{15.365 } & \textbf{0.621 } & \textbf{2.011 } & \textbf{2.954 } & \textbf{12.440 } & \textbf{8.013 } & \textbf{0.720 } & \textbf{3.848 } & \textbf{2.913 } & \textbf{12.991 } & \textbf{8.844 } & \textbf{0.704 } & \textbf{2.886 } \\
            \bottomrule
        \end{tabular}
    }
    \label{tab:performance}
    \vspace{-5pt}
\end{table*}%
% This chapter presents the recommendation and debiasing performance of IPL on four datasets and verifies its effectiveness in  simultaneously improving both debiasing and recommendation effectiveness. Furthermore, a case study is conducted to explain the rationale behind the win-win situation.

\subsection{Experiment Setup}
\textbf{Datasets.} We conduct experiments on four public benchmark datasets: \textbf{MovieLens-1M}~\cite{harper2015movielens}, \textbf{Gowalla}~\cite{cho2011friendship}, \textbf{Yelp}~\cite{yelpdataset}, and \textbf{Amazon Book}~\cite{he2016ups}.
We uniformly sample interaction records from each item, splitting them into 70\%, 10\%, and 20\% as train, validation, and test sets.
In addition, we follow~\cite{yang2018unbiased} and use MF~\cite{koren2009matrix} to estimate $\gamma$ of the datasets, which are $1.826$, $1.285$, $1.552$, $1.446$, respectively.

\noindent\textbf{Baselines.} We implement our method with MF \cite{koren2009matrix} and LightGCN \cite{he2020lightgcn} \footnote{The model implementations are accessible on Github: https://github.com/UnitDan/IPL}, comparing with the following baselines: 
% \textbf{IPW}~\cite{liang2016causal} is a classic re-weighting method that enhances the weight of unpopular items. \textbf{PO}~\cite{zhu2021popularity} is a regularization method that penalizes the correlation between item popularity and its ranking. 
\textbf{IPW}~\cite{liang2016causal} is a IPS-based re-weighting method.
\textbf{PO}~\cite{zhu2021popularity} eliminates popularity bias in item ranking by regularization.
\textbf{MACR} \cite{wei2021model} removes the causal effect of popularity on preference scores. 
\textbf{DPR}~\cite{zhu2020measuring} uses adversarial training for item group fairness.
In addition, we also optimize \textbf{MF} and \textbf{LightGCN} using BPR loss as the basis for comparison.

\noindent\textbf{Evaluation metrics.} We evaluate recommendation performance by \textit{Precision@k}, \textit{Recall@k} and \textit{NDCG@k}.  We set $k=20$ following \cite{wang2019neural}, and show all these metrics after multiplying by 100~\cite{wu2022inmo}. To evaluate the effectiveness of popularity debiasing, we use the mutual information (\textit{MI})~\cite{cover1999elements} of $r_i$ and item popularity $Q_i^*$ to measure the correlation between them. We also propose a metric that is consistent with the IPL criterion, called \textit{Deviation of Interaction distribution} (\textit{DI}). $DI = \frac{std_{i\in\mathcal{I}}(r_i)}{mean_{i\in\mathcal{I}}(r_i)}$.
Lower values of MI and DI indicate less popularity bias in the recommendation result.

\vspace{-5pt}
\subsection{Performance Evaluation}

To demonstrate the effectiveness of our method, we evaluate recommendation performance and debiasing results on four benchmark datasets. As shown in Table~\ref{tab:performance}, MF and LightGCN achieve good recommendation performance, but with a serious popularity bias, reflected in a high value of MI and DI. For debiasing baselines, they mitigate the popularity bias to some extent, however, at the cost of severely sacrificing recommendation performance. For our method, IPL shows superiority in both recommendation accuracy and debiasing. Especially, IPL achieves the best recommendation performance on all datasets, even outperforming MF and LightGCN. As for debiasing, our IPL performs best on MI and DI. These results demonstrate the effectiveness of our IPL method in terms of both recommendation performance and popularity debiasing.

\vspace{-5pt}

\subsection{IPL Method Win-Win Evaluation}

To verify the conclusion in Section~\ref{theoretical} that our proposed method can achieve a win-win situation in recommendation and debiasing, we compared the performance of IPL with different strengths of popularity debiasing and basic models. We used $1/DI$ to represent the debiasing performance and an unbiased recall metric ($Recall_{\text{SNIPS}}$)~\cite{yang2018unbiased} to represent the recommendation accuracy. We experimented with 20 equally spaced values of $\lambda_f$ on a logarithmic scale from $1e-2$ to $1e-6$. Fig.~\ref{fig:trade-off} demonstrates that the experimental results of IPL are consistently located in the upper right space of the basic model results, indicating that IPL can simultaneously enhance recommendation and debiasing performance.

\begin{figure}
    \centering
    \includegraphics[width=7.9cm]{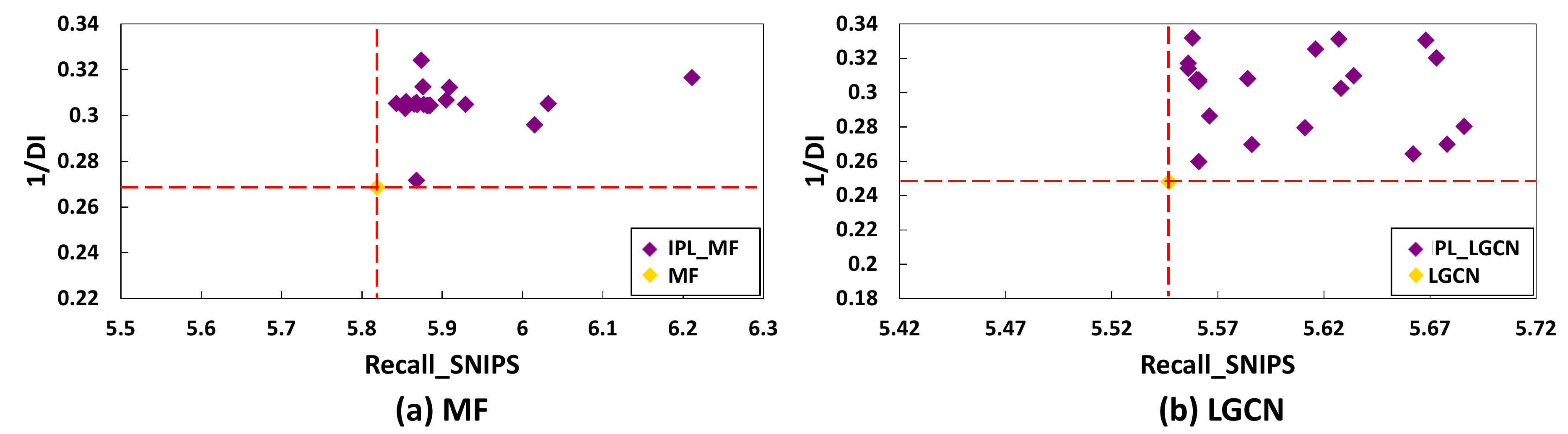}
    \setlength{\abovecaptionskip}{0.cm}
    \caption{Performance of basic models and IPL with different $\lambda_f$ on recommendation and debiasing.}
    \label{fig:trade-off}
    \vspace{-5pt}
\end{figure}

\vspace{-5pt}
\subsection{Case Study}
\begin{figure}
    \centering
    \includegraphics[width=7.5cm]{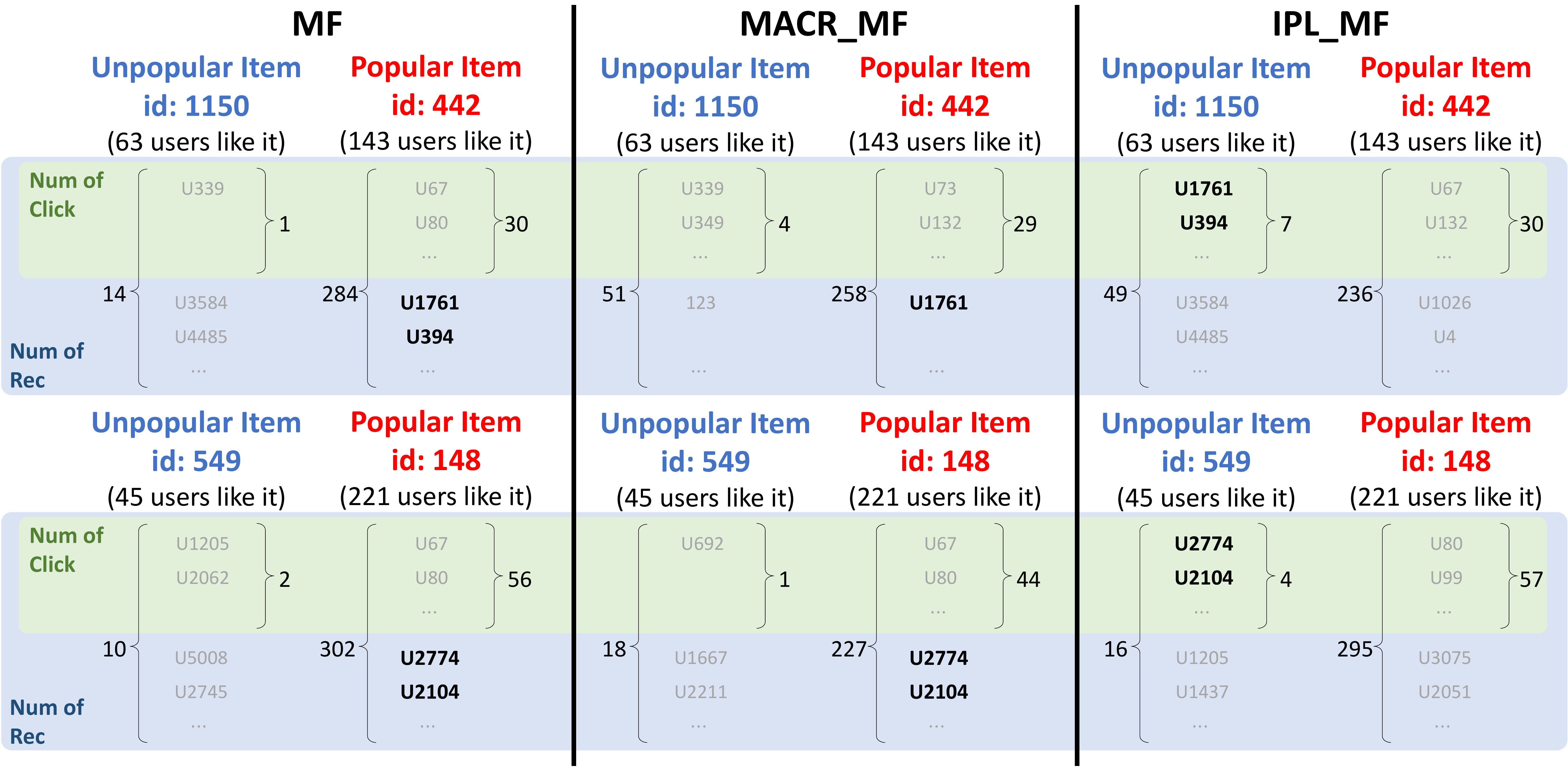}
    \setlength{\abovecaptionskip}{0.cm}
    \caption{Case study. The number of recommendations (blue) is on the left, and that of effective recommendations (green) is on the right.  Users who prefer unpopular items are bold.}
    \label{fig:case}
    \vspace{-8pt}
\end{figure}

Fig.~\ref{fig:case} demonstrates how IPL can achieve a win-win situation using Movie-Lens as an example. By comparing recommendation lists of two pairs of popular and unpopular items, we found that MF amplifies popularity bias by over-recommending popular items (442 and 148) and under-recommending unpopular ones (1150 and 549). The existing debiasing method (MACR\_MF) blindly recommends more to unpopular items, sacrificing recommendations of popular items. In contrast, our IPL\_MF recommends the unpopular item (1150) to users (U1761 and U394) who prefer unpopular items, thereby alleviating the popularity bias while achieving better recommendation accuracy.

\vspace{-5pt}
\section{CONCLUSION}

In this work, we propose a new criterion to measure popularity bias, i.e., an item should get interactions proportional to the number of users who like it (IPL). We further propose a debiasing framework based on IPL.  Both theoretical analysis and experimental results demonstrate that the proposed method can achieve a win-win situation for both recommendation boosting and popularity debiasing.
% A future direction is to investigate the underlying mechanisms of popularity bias and debiasing methods.

\begin{acks}
This work is funded by the National Key R\&D Program of China (2022YFB3103700, 2022YFB3103701), and the National Natural Science Foundation of China under Grant Nos. 62272125, U21B2046. Huawei Shen is also supported by Beijing Academy of Artificial Intelligence (BAAI).
\end{acks}

\bibliographystyle{ACM-Reference-Format}
\balance
\bibliography{ref}

\end{document}